\documentclass[a4paper]{article}

\usepackage{
amsmath,
amsthm,
amscd,
amssymb,
}
\usepackage{listings}

\usepackage{xint}
\def\allowsplits #1{\ifx #1\relax \else #1\hskip 0pt plus 1pt\relax
\expandafter\allowsplits\fi}%
\def\printnumber #1{\expandafter\allowsplits \romannumeral-`0#1\relax }%

\usepackage{comment}
\usepackage{tikz}
\usetikzlibrary{positioning,arrows,calc}
\usepackage{xspace}

\usepackage{graphicx}

\usepackage{url}

\theoremstyle{plain}
\newtheorem{lemma}{Lemma}
\newtheorem{definition}{Definition}
\newtheorem{corollary}{Corollary}
\newtheorem{proposition}{Proposition}
\newtheorem{theorem}{Theorem}

\newtheorem{remark}{Remark}

\newtheoremstyle{derp}
{3pt}
{3pt}
{}
{}
{\upshape}
{:}
{.5em}
{}
\theoremstyle{derp}
\newtheorem{example}{Example}

\newcommand{\Z}{\mathbb{Z}}

\newcommand{\N}{\mathbb{N}}

\newcommand{\B}{\mathcal{L}}

\newcommand{\INF}{{}^\infty}
\newcommand{\ID}{\mathrm{id}}

\newcommand\xqed[1]{%
  \leavevmode\unskip\penalty9999 \hbox{}\nobreak\hfill
  \quad\hbox{#1}}
\newcommand\qee{\xqed{$\triangle$}}

\newcommand{\End}{\mathrm{End}}

\title{Von Neumann regularity, split epicness and elementary cellular automata}

\author{
Ville Salo \\
vosalo@utu.fi
}

\begin{document}
\maketitle

\begin{abstract}
We show that a cellular automaton on a mixing subshift of finite type is a Von Neumann regular element in the semigroup of cellular automata if and only if it is split epic onto its image in the category of sofic shifts and block maps. It follows from [S.-T\"orm\"a, 2015] that Von Neumann regularity is decidable condition, and we decide it for all elementary CA.
\end{abstract}

\section{Introduction}


The Von Neumann regular elements -- elements $a$ having a weak inverse $b$ such that $aba = a$ -- of cellular automaton (CA) semigroups are studied in \cite{CaGa18}. We show that in the context of cellular automata on mixing subshifts of finite type, Von Neumann regularity coincides with the notion of split epicness onto the image, another generalized invertibility notion from category theory.

Question 1 of \cite{CaGa18} asks which of the so-called elementary cellular automata (ECA) are Von Neumann regular. They determine this for all ECA except (ones isomorphic to) those with numbers 6, 7, 9, 23, 27, 28, 33, 41, 57, 58 and 77, see the next section for the definition of the numbering scheme. 

What makes this question interesting is that Von Neumann regularity is not a priori decidable -- clearly checking if $g$ is a weak inverse is semidecidable, but it is not immediately clear how to semidecide the nonexistence of a weak inverse. However, split epicness has been studied previously in \cite{SaTo15a}, and in particular it was shown there that split epicness of a morphism between two sofic shifts is a decidable condition. This means Question 1 of \cite{CaGa18} can in theory be decided algorithmically.

As the actual bound stated in \cite{SaTo15a} is beyond astronomical, it is an interesting question whether the method succeeds in actually deciding each case. With a combination of this method, computer and manual searches, and some ad hoc tricks, we prove that ECA 6, 7, 23, 33, 57 and 77 are Von Neumann regular, while 9, 27, 28, 41 and 58 are not, answering the remaining cases of Question 1 of \cite{CaGa18}.

The Von Neumann regular CA on this list have weak inverses of radius at most five. Non-regularity is proved in each case by looking at eventually periodic points of eventual period one. The non-regularity of all but ECA 9 and ECA 28 can be proved by simply observing that their images are proper sofic, though we also explain why they are not regular using the method of \cite{SaTo15a}.

\section{Preliminaries}

The \emph{full shift} is $\Sigma^\Z$ where $\Sigma$ is a finite alphabet, which is a dynamical system under the shift $\sigma(x)_i = x_{i+1}$. Its subsystems (closed shift-invariant subsets) are called \emph{subshifts}. A \emph{cellular automaton (CA)} is a shift-commuting continuous function $f : X \to X$ on a subshift $X$. The cellular automata on a subshift $X$ form a monoid $\End(X)$. A CA $f$ is \emph{reversible} if $\exists g: f \circ g = g \circ f = \ID$.

A cellular automaton has a local rule, that is, there exists a \emph{radius} $r \in \N$ such that $f(x)_i$ is determined by $x|_{[i-r,i+r]}$ for all $x \in X$ (and does not depend on $i$). The \emph{elementary cellular automata} (ECA) are the CA on the binary full shift $\{0,1\}^\Z$ which can be defined with radius $1$. There is a numbering scheme for such CA: If $n \in [0,255]$ has base $2$ representation $b_7b_6...b_1b_0$, then ECA number $n$ is the one mapping
\[ f(x)_i = 1 \iff b_{(x_{[i-1,i+1]})_2} = 1 \]
where $(x_{[i-1,i+1]})_2$ is the number $x_{[i-1,i+1]}$ represents in base $2$.
This numbering scheme is from \cite{Wo83}.

The usage of base 10 in this notation is standard, and many CA researchers remember ECA by these numbers. However, for clarity we switch to hexadecimal notation from radius $2$ upward.

A subshift can be defined by forbidding a set of words from appearing, and this is in fact a characterization of subshifts. A subshift is \emph{of finite type} or \emph{SFT} if it can be defined by a finite set of forbidden patterns, and \emph{sofic} if it can be defined by a forbidden regular language in the sense of automata and formal languages.

The \emph{language} of a subshift is the set of words that appear in its configurations. A subshift $X$ is \emph{mixing} if for all words $u, v$ appearing in the language of $X$, for all large enough $n$ some word $uwv$ with $|w| = n$ appears in the language of $X$.

See standard references for more information on symbolic dynamics \cite{LiMa95} or automata theory \cite{HoMoUl06}.

\section{Split epicness and Von Neumann regularity}

In this section, we show split epicness and Von Neumann regularity are equal concepts on mixing SFTs. On the full shift, this is simply a matter of defining these terms.

If $S$ is a semigroup, then $a \in S$ is \emph{(Von Neumann) regular} if $\exists b \in S: a b a = a \wedge b a b = b$. We say $b$ is a \emph{generalized inverse} of $a$. If $aba = a$, then $b$ is a \emph{weak generalized inverse} of $a$.

\begin{lemma}
If $a$ has a weak generalized inverse, then it has a generalized inverse and thus is regular
\end{lemma}

\begin{proof}
If $aba = a$, then letting $c = bab$, we have
$aca = ababa = aba = a$
and
$cac = bababab = babab = bab = c$.
\end{proof}

If $\mathcal{C}$ is a category, a morphism $f : X \to Y$ is \emph{split epic} if there is a morphism $g : Y \to X$ such that $f \circ g = \ID_Y$. Such a $g$ is called a \emph{right inverse} or a \emph{section}.

Note that in general category-theoretic concepts depend on the particular category at hand, but if $\mathcal{C}$ is a full subcategory of $\mathcal{D}$ (meaning a subcategory induced by a subclass of the objects, by taking all the morphisms between them), then split epicness for a morphism $f : X \to Y$ where $X, Y$ are objects of $\mathcal{C}$ means the same in both.

We are in particular interested in the category $K3$ (in the naming scheme of \cite{SaTo15a}) with sofic shifts as objects, and \emph{block maps}, i.e. shift-commuting continuous functions $f : X \to Y$ as morphisms.

The following theorem is essentially only a matter of translating terminology, and works in many concrete categories.

\begin{theorem}
Let $X$ be a sofic shift, and $f : X \to X$ a cellular automaton. Then the following are equivalent:
\begin{itemize}
\item $f : X \to f(X)$ has a right inverse $g : f(X) \to X$ which can be extended to a morphism $h : X \to X$ such that $h|_{f(X)} = g$,
\item $f$ is regular as an element of $\End(X)$.
\end{itemize}
\end{theorem}

\begin{proof}
Suppose first that $f$ is regular, and $h \in \End(X)$ satisfies $fhf = f$ and $hfh = h$. Then the restriction $g = h|_{f(X)} : f(X) \to X$ is still shift-commuting and continuous, and $\forall x: fg(f(x)) = f(x)$ implies that for all $y \in f(X)$, $fg(y) = y$, i.e. $g$ is a right inverse for the codomain restriction $f : X \to f(X)$ and it extends to the map $h : X \to X$ by definition.

Suppose then that $fg = \ID_{f(X)}$ for some $g : f(X) \to X$, as a right inverse of the codomain restriction $f : X \to f(X)$. Let $h : X \to X$ be such that $h|_{f(X)} = g$, which exists by assumption. Then $fh(f(x)) = fg(f(x)) = f(x)$. Thus $f$ is regular, and that $hfh$ is a generalized inverse for it.
\end{proof}

Note that when $X$ is a full shift, extending morphisms is trivial: simply fill in the local rule arbitrarily. The Extension Lemma generalizes this idea:

\begin{theorem}
Let $X$ be a mixing SFT, and $f : X \to X$ a cellular automaton. Then the following are equivalent:
\begin{itemize}
\item $f : X \to f(X)$ is split epic in $K3$.
\item $f$ is regular as an element of $\End(X)$.
\end{itemize}
\end{theorem}

\begin{proof}
It is enough to show that any right inverse $g : f(X) \to X$ can be extended to $h : X \to X$ such that $h|_{f(X)} = g$. By the Extension Lemma \cite{LiMa95}, it is enough to show that $X \searrow X$, i.e. for every point $x \in X$ with minimal period $p$, there is a point $y \in X$ with minimal period dividing $p$. This holds trivially.
\end{proof}

\section{Deciding split epicness}

We recall the characterization of split epicness \cite[Theorem~1]{SaTo15a}. This is Theorem~\ref{thm:SplitEpicSolved} below.

\begin{definition}
\label{def:SPP}
Let $X, Y$ be subshifts and let $f : X \to Y$ be a morphism. Define
\[ \mathcal{P}_p(Y) = \{ u \in \B(Y) \;|\; \INF u \INF \in Y, |u| \leq p \}. \]
We say $f$ satisfies the \emph{strong $p$-periodic point condition} if there exists a length-preserving function $G : \mathcal{P}_p(Y) \to \B(X)$ such that for all $u, v \in \mathcal{P}_p(Y)$ and $w \in \B(Y)$ with $\INF u .w v \INF \in Y$, there exists an $f$-preimage for $\INF u .w v \INF$ of the form
$\INF G(u) w' . w'' w''' G(v) \INF \in X$
where $|u|$ divides $|w'|$, $|v|$ divides $|w'''|$ and $|w| = |w''|$. The \emph{strong periodic point condition} is that the strong $p$-periodic point condition holds for all $p \in \N$.
\end{definition}

The strong periodic point condition is an obvious necessarily condition for having a right inverse, as the right inverse must consistently pick preimages for periodic points, and they must satisfy these properties; \cite[Theorem~1]{SaTo15a}, proves that it is also sufficient in the case when $X$ is an SFT and $Y$ is a sofic shift.

\begin{theorem}
\label{thm:SplitEpicSolved}
Given two objects $X \subset S^\Z$ and $Y \subset R^\Z$ and a morphism $f : X \to Y$ in K3, it is decidable whether $f$ is split epic. If $X$ is an SFT, split epicness is equivalent to the strong periodic point condition.
\end{theorem}

We note that Definition~\ref{def:SPP} is equivalent to a variant of it where $G$ is only defined on Lyndon words \cite{Lo02}, i.e. lexicographically minimal representative words of periodic orbits: if $G$ is defined on those, it can be extended to all of $\mathcal{P}_p$ in an obvious way, and the condition being satisfied by minimal representatives implies it for all eventually periodic points.

\begin{remark}
It is observed in \cite[Theorem~1]{CaGa18} that if $f : X \to Y$ is split epic, then every periodic point in $Y$ must have a preimage of the same period in $X$ -- this is a special case of the above, and could thus be called the \emph{weak periodic point condition}. In \cite[Example~5]{SaTo15a}, an example is given of morphism between mixing SFTs which satisfies the weak periodic point condition but not the strong one. We have not attempted to construct an example of a CA on a full shift which has this property onto its image, and we did not check whether any non-regular ECA satisfies it. In \cite[Theorem~4]{CaGa18}, for full shifts on finite groups, the weak periodic point condition is shown to be equivalent to split epicness (when CA are considered to be morphisms onto their image). In the context of CA on $\Z^2$, there is no useful strong periodic point condition in the sense that split epicness is undecidable, see Corollary~\ref{cor:2D}. 
\end{remark}

In the proof of Theorem~\ref{thm:SplitEpicSolved} in \cite{SaTo15a}, decidability is obtained from giving a bound on the radius of a minimal inverse, and a very large one is given, as we were only interested in the theoretical decidability result. The method is, however, quite reasonable in practise:
\begin{itemize}
\item To semidecide non-(split epicness), look at periodic points one by one, and try out different possible choices for their preimages. Check by automata-theoretic methods (or ``by inspection'') which of these are consistent in the sense of Definition~\ref{def:SPP}.
\item To semidecide split epicness, invent a right inverse -- note that here we can use the other semialgorithm (running in parallel) as a tool, as it tells us more and more information about how the right inverse must behave on periodic points, which tells us more and more values of the local rule.
\end{itemize}
One of these is guaranteed to finish eventually by \cite{SaTo15a}.

Proposition~\ref{prop:soficnotsplit} below is a slight generalization of Proposition~1 in \cite{SaTo15a}. We give a proof here, as the proof in \cite{SaTo15a} unnecessarily applies a more difficult result of S. Taati (and thus needs the additional assumption of ``mixing''). This Proposition allows us to obtain non-regularity of all of the non-regular ECA considered here apart from ECA 9 and 28, though we do also provide a strong periodic point condition argument for all the non-regular ECA.

\begin{lemma}
If $X$ is an SFT and $f : X \to X$ is idempotent, i.e. $f^2 = f$, then $f(X)$ is an SFT.
\end{lemma}

\begin{proof}
Clearly $x \in f(X) \iff f(x) = x$, which is an SFT condition.
\end{proof}

\begin{proposition}
\label{prop:soficnotsplit}
If $X$ is an SFT and $f : X \to X$ is regular, then $f(X)$ is of finite type.
\end{proposition}

\begin{proof}
Let $g : X \to X$ be a weak inverse. Then $g \circ f : X \to X$ is idempotent, so $g(f(X))$ is an SFT. Note that the domain-codomain restriction $g|_{f(X), g(f(X))} : f(X) \to g(f(X))$ is a conjugacy between $f(X)$ and $g(f(X))$: its two-sided inverse is $f|_{g(f(X))} \to f(X)$ by a direct computation. Thus $f(X)$ is also an SFT.
\end{proof}

We also mention another condition, although it is not applicable in the proofs.

\begin{lemma}
\label{lem:InjImpliesSurj}
Let $X$ be a subshift with dense periodic points and $f : X \to X$ a cellular automaton. If $f$ is injective, it is surjective.
\end{lemma}

\begin{proof}
The set $X_p = \{x \in X \;|\; \sigma_p(x) = x\}$ satisfies $f(X_p) \subset X_p$. Since $f$ is injective and $X_p$ is finite, we must have $f(X_p) = X_p$. Thus $f(X)$ is a closed set containing the periodic points. If periodic points are dense, $f(X) = X$.
\end{proof}

We are interested mainly in mixing SFTs, where periodic points are easily seen to be dense. In the case of mixing SFTs, the previous lemma can also be proved with an entropy argument: An injective CA cannot have a diamond when seen as a block map, so \cite[Theorem~8.1.16]{LiMa95} shows that the entropy of the image $f(X)$ of an injective CA is equal to the entropy of $X$. By \cite[Corollary~4.4.9]{LiMa95}, $X$ is \emph{entropy minimal}, that is, has no proper subshifts of the same entropy, and it follows that $f(X) = X$.

\begin{proposition}
\label{prop:SurjectiveNotInjective}
Let $X$ be a mixing SFT and $f : X \to X$ a surjective cellular automaton. Then $f$ is injective if and only if it is regular.
\end{proposition}

\begin{proof}
By the previous lemma, an injective CA on a mixing SFT is surjective, thus bijective, thus reversible, thus regular. Conversely, let $f$ be surjective and regular, then let $g : X \to X$ be a weak generalized inverse. Then $g$ is injective, so it is surjective, thus reversible. Thus $f$ must be bijective as well.
\end{proof}

More generally, the previous proposition works on surjunctive subshifts in the sense of \cite[Exercise~3.29]{CeCo10}, i.e., subshifts where injective cellular automata are surjective. In particular this is the case for full shifts on surjunctive groups such as abelian ones. Since injectivity is undecidable for surjective CA on $\Z^d$, $d \geq 2$ by \cite{Ka94}, we obtain the following corollary.

\begin{corollary}
\label{cor:2D}
Given a surjective CA $f : \Sigma^{\Z^2} \to \Sigma^{\Z^2}$, it is undecidable whether $f$ is split epic.
\end{corollary}

\section{Von Neumann regularity of elementary CA}



\begin{theorem}
The elementary CA with numbers 6, 7, 23, 33, 57 and 77 are regular.
\end{theorem}

\begin{proof}
It is a finite case analysis to verify that the CA defined in Figure~\ref{fig:rule6inverse}, Figure~\ref{fig:rule7inverse}, Figure~\ref{fig:rule23inverse}, Figure~\ref{fig:rule33inverse}, Figure~\ref{fig:rule57inverse} and Figure~\ref{fig:rule77inverse} in the appendix are generalized inverses of the respective ECA.

Two Python\footnote{Tested with Python 2.7.13 and Python 3.6.2.} scripts that perform a computer verification are included: Appendix~\ref{sec:Code} contains code for checking that the graphical rules we list in the figures define correct weak inverses. Appendix~\ref{sec:CodeNum} contains code for checking that the numbers we list for the rules also define correct weak inverses.
\end{proof}

Before the non-regularity proofs, let us show the XOR CA with neighborhood $\{0,1\}$ is not regular -- this is clear from the fact it is surjective, and from the fact there are $1$-periodic points with no inverse of period $1$, but it also neatly illustrates the strong periodic point method.

\begin{example}
\label{ex:XOR}
The  CA $f : \{0,1\}^\Z \to \{0,1\}^\Z$ defined by
\[ f(x)_i = 1 \iff x_i + x_{i+1} \equiv 1 \bmod 2 \]
is not regular. To see this, consider the strong periodic point condition for $p = 1$. Since $f(0^\Z) = f(1^\Z) = 0^\Z$, the point $0^\Z$ has two preimages. It is enough to show that neither choice $z \in \{0^\Z, 1^\Z\}$ is consistent, i.e. there is a point $y$ which is in the image of $f$ such that $y$ has no preimage that is left and right asymptotic to $z$. This is shown by considering the point
\[ y = ...0000001000000... \]
(which is in the image of $f$ since $f$ is surjective) and observing that no preimage of it is asymptotic to the same point $z \in \{0^\Z, 1^\Z\}$ -- if a preimage is left-asymptotic to $a^\Z$, it is right-asymptotic to $(1-a)^\Z$. \qee
\end{example}

\begin{theorem}
The elementary CA with numbers 9, 27, 28, 41 and 58 are not regular.
\end{theorem}

\begin{proof}
See the lemmas below.
\end{proof}

\begin{lemma}
\label{lem:9}
The elementary CA 9 is not regular.
\end{lemma}

\begin{proof}
Let $f$ be the ECA $9$, i.e. $f(x)_i = 1 \iff x_{[i-1,i+1]} \in \{000, 011\}$. The image $X$ of $f$ is the SFT with forbidden patterns $1011$, $10101$, $11001$, $11000011$ and $110000101$. One can verify\footnote{For verifying only the proof of this lemma, i.e. the non-regularity of ECA 9, it is enough to show that the point $x$ below is in $X$, that is, it has some preimage ($...0100100001001001...$ is one). Knowing the SFT is, however, essential for finding such an argument, so we argue in this way, again to illustrate the method.} this with standard automata-theoretic methods.


We have $f(0^\Z) = 1^\Z$ and $f(1^\Z) = 0^\Z$, so if $g : X \to \{0,1\}^\Z$ is a right inverse for $f$, then $g(0^\Z) = 1^\Z$. Consider now the configuration
\[ x = ...0000011.00000... \in X \]
where coordinate $0$ is to the left of the decimal point (i.e. the rightmost $1$ or the word $11$). Let $g(x) = y$. Then $y_i = 1$ for all large enough $i$ and $y_i = 0$ for some $i$. Let $n$ be maximal such that $y_n = 0$. Then $y_{[n,n+2]} = 011$ so $f(y)_{n+1} = 1$ and $f(y)_{n+1+i} = 0$ for all $i \geq 1$. Since $f(y) = x$, we must have $n = -1$ and since $\{000, 011\}$ does not contain a word of the form $a01$, it follows that $f(y)_{-1} = 0 \neq x_{-1}$, a contradiction.
\end{proof}

The proof shows that the CA does not have the strong periodic point property for $p = 1$. In general, for fixed $p$ one can use automata-theory to decide whether it holds up to that $p$, though here (and in all other proofs) we found the contradictions by hand before we had to worry about actually implementing this.

\begin{lemma}
The ECA $27$ is not regular.
\end{lemma}


\begin{proof}
Let $f$ be the ECA $27$, i.e. $f(x)_i = 1 \iff x_{[i-1,i+1]} \in \{000, 001, 011, 100\}$. The image $X$ of $f$ is proper sofic, we omit the automaton and argue directly in terms of configurations. Proposition~\ref{prop:soficnotsplit} directly shows that the CA can not be regular in the case when the image is proper sofic, but we give a direct proof to illustrate the method (and so that we do not have to provide a proof that the image is sofic, which is straightforward but lengthy).

Again, we will see that this CA does not satisfy the strong periodic point condition for $p = 1$. Observe that $f(1^\Z) = 0^\Z$ and $f(0^\Z) = 1^\Z$ so if $g$ is a right inverse from the image to $\{0,1\}^\Z$, then $g(0^\Z) = 1^\Z$ and $g(1^\Z) = 0^\Z$. Let $y = ...000001100.10101010...$ and observe that
\begin{align*}
f(y) = f(...&000001100.10101010...) = \\
         ...&111111011.00000000... = x \in X.
\end{align*}

We now reason similarly as in Lemma~\ref{lem:9}. We have $g(x)_i = 1$ for all large enough $i$, and if is $n$ is maximal such that $g(x)_n = 0$, then $f(g(x))_{n+1} = 1$ and $f(g(x))_{n+1+i} = 0$ for all $i \geq 1$, so again necessarily $n = -1$. A short combinatorial analysis shows that no continuation to the left from $n$ produces $f(g(x))_{n} = 1$ and $f(g(x))_{n-1} = 0$, that is, the image of $g$ has no possible continuation up to coordinate $-1$.
\end{proof}

\begin{lemma}
The ECA $28$ is not regular.
\end{lemma}


\begin{proof}
Let $f$ be the ECA $28$, i.e. $f(x)_i = 1 \iff x_{[i-1,i+1]} \in \{010, 011, 100\}$. The image $X$ of $f$ is the SFT with the single forbidden pattern $111$.


We have $f(0^\Z) = f(1^\Z) = 0^\Z$. The point
\[ ...000010000... \in X \]
contradicts the choice $g(0^\Z) = 0^\Z$ by a similar analysis as in previous theorems; similarly as in Example~\ref{ex:XOR}, computing the preimage from right to left, the asymptotic type necessarily changes. Thus we must have $g(0^\Z) = 1^\Z$.

On the other hand, if $g(0^\Z) = 1^\Z$, then going from right to left, we cannot find a preimage for
\[ ...000110000... \in X. \]
(Alternatively, going from left to right, the asymptotic type necessarily changes to $0$s or never becomes $1$-periodic.)

It follows that $g(0^\Z)$ has no consistent possible choice, a contradiction.
\end{proof}

\begin{lemma}
The ECA $41$ is not regular.
\end{lemma}


\begin{proof}
Let $f$ be the ECA $41$, i.e. $f(x)_i = 1 \iff x_{[i-1,i+1]} \in \{000, 011, 101\}$. The image $X$ of $f$ is proper sofic, we omit the automaton and argue directly in terms of configurations. Again Proposition~\ref{prop:soficnotsplit} would also yield the result.

Again, we will see that this CA does not satisfy the strong periodic point condition for $p = 1$. Observe that $f(1^\Z) = 0^\Z$ and $f(0^\Z) = 1^\Z$ so if $g$ is a right inverse from the image to $\{0,1\}^\Z$, then $g(0^\Z) = 1^\Z$ and $g(1^\Z) = 0^\Z$. Let $y = ...0000000100100100...$ so
\begin{align*}
f(y) = f(...&0000000100100100...) = \\
         ...&1111110000000000... = x \in X.
\end{align*}
In the usual way (right to left), we verify that $x$ has no preimage that is right asymptotic to $1^\Z$, obtaining a contradiction.
\end{proof}

\begin{lemma}
The ECA $58$ is not regular.
\end{lemma}


\begin{proof}
Let $f$ be the ECA $58$, i.e. $f(x)_i = 1 \iff x_{[i-1,i+1]} \in \{001, 011, 100, 101\}$. The image $X$ of $f$ is proper sofic, we omit the automaton. Again Proposition~\ref{prop:soficnotsplit} would also yield the result.

The point $0^\Z$ has two $1$-periodic preimages. We show neither choice satisfies the strong periodic point condition: if $g(0^\Z) = 1^\Z$, then $g$ cannot give a preimage for
\[ ...000000010000000... \]
If $g(0^\Z) = 0^\Z$, then it cannot give a preimage for
\[ ...000000011000000... \]
It is easy to find preimages for these two configurations, however, so ECA 58 is not regular.
\end{proof}

%




\section{Notes}
\label{sec:notes}

\subsection{How were the inverses found?}

We did not know what to expect when we embarked on this journey (apart from the fact that a certain tower of exponentials we did not dare evaluate is an upper bound on the radius of the inverses). One could probably find the inverses more directly now that we know that they are very small. Here, we explain the messy process by which we found them. This section is not important for verifying the results.


We worked through the ECA listed in \cite{CaGa18} in increasing numerical order, so we began with ECA 6. Let $f : \{0,1\}^\Z \to \{0,1\}^\Z$ be ECA $6$. It seemed immediately clear that it is easier to work with split epicness than regularity: computationally, this is easier since we never have to compute the large-radius CA $f \circ g \circ f$, and for proofs it was easier since the theory of \cite{SaTo15a} was developed in this setting.

Thus, the first step was to compute the image of $f$. It is an exercise in symbolic dynamical automata theory to check that $X = f(\{0,1\}^\Z)$ is the subshift of finite type with forbidden patterns $111, 100101, 1001001, 11000101, 110001001$. We computed this by finding a minimal DFA for the image using SAGE \cite{Sage}, computed its regular language of minimal forbidden patterns using closure properties, and found the minimal automaton is acyclic and accepts the given list of finite words. Now, it was enough to define $g$ on words that do not contain these forbidden patterns, and verify that $f \circ g$ behaves as the identity function this SFT. 

At this point, we manually applied the algorithm in \cite{SaTo15a} to show that if $g : X \to \{0,1\}^\Z$ is a right inverse for $f : \{0,1\}^\Z \to X$, then $g(0^\Z) = 1^\Z$. Namely, with the a priori possible choice $g(0^\Z) = 0^\Z$, the strong periodic point condition fails: the eventually periodic point $...000101000...$ has no preimage consistent with this choice. No other restrictions were found at this point, and we quickly convinced ourselves that there indeed was a shift-invariant continuous procedure for finding inverses. However, proving that it works in all cases (and even fully explaining the rule) sounded like a daunting task, so we turned back to the computer.

We guessed from experimenting that radius $5$ is sufficient, and we guessed that it is safe to force $g(x.0000y) = x'.a11by'$ for any tails $x, y$, since $g(0^\Z) = 1^\Z$ was guaranteed.

We then implemented basic ``sudoku deduction rules'' on the computer to determine the possible symbol that can be seen in a particular cell of a preimage configuration -- the computer looked through all legal words of length $11$, finding the possible symbols that can appear on top of a particular symbol (starting with the ``seed'' rule $0000 \mapsto 11$). The author then looked at configurations that were not yet assigned a preimage, and followed the intuitive procedure to invent suitable preimages (and the definition was automatically extended by the sudoku rule). First only provably ``safe'' choices were made by proving that a particular choice of preimage symbol is more general than another, and in the end a few ad hoc rules were introduced by guesswork.

Once the rule was specified for all words of length 11, it turned out to indeed be a right inverse for $f$, which was easy to verify by computer. We then carried out another joint effort by human and computer to simplify the presentation of $g$ to the rules in Figure~\ref{fig:rule6inverse}. This presentation is a local minimum among such presentations, in the sense that one cannot remove any of the individual rules or bits in the presentation of the rule. Presumably the rule has a simpler representation globally, and note that only its values on the image subshift of $f$ actually matter. We did not check whether it has a weak inverse of radius $3$.

For other rules, we proceeded similarly, and were able to conclude all cases except rule 57. For rules 33 and 77, after finding them we manually added some symmetry to the presentations of the weak inverses, to make them slightly more redundant but perhaps easier to remember and apply. For ECA $7$ we originally found an inverse of radius $4$, but while proving that ECA 6 does not have an inverse of radius $2$ we found (fully with pen and paper) that ECA 7 in fact has a unique inverse of radius $2$ (which is essentially an elementary CA), which we report here.

The inverse for CA 57 was more difficult one to find. To find it, we used the same procedure we used to prove rule $6$ has no inverse of radius $2$, whose strength lies somewhere between the observation of \cite{CaGa18} that periodic points should have preimages of the same period, and the full strong periodic point condition of \cite{SaTo15a}: We computed periodic points with periods $p = 1, ..., 20$, and for each point, we looked through its preimages of the same period, and if all of them had the same bit $b$ in a particular position, we knew the preimage of this periodic point had to have $b$ in that position, which forced the image of the word in the corresponding point of the orbit of that periodic point. This simple procedure picked, after a few minutes of runtime (with my highly non-optimized implementation in Python), the images of all but 36 of the words of length $7$. These words were, not surprisingly, of high period, and thus could occur next to each other in only a few different ways, so it was easy to manually pick images for them yielding an inverse.

According to \cite{CaGa18}, none of the ECA 6, 7, 23, 33, 57 and 77 have an elementary weak inverse, making our inverses optimal in terms of radius apart from ECA 6 and 57.

The ECA 6 does not have a weak inverse of radius $2$ by a simple proof: Let $f$ be ECA 6 and suppose $g$ is a weak inverse of radius $2$. We have $f((1100)^\Z) = (0001)^\Z$ and the point $(0001)^\Z$ has no other $f$-preimage of period $4$. We have $f((00001)^\Z) = (00011)^\Z$ and the point $(00011)^\Z$ has no other $f$-preimage of period $5$. It follows that $g((0001)^\Z) = (1100)^\Z$ and $g((00011)^\Z) = (00001)^\Z$, which contradicts $g$ having radius $2$ since the local rule of $g$ must map the word $10001$ to both $0$ and $1$. We did not attempt to find a radius $3$ inverse.

For rule 57, the computer reported a contradiction (double-definition of the local rule) with radius $3$ when applying the above-mentioned procedure, so we believe the current radius is optimal. However, this should not be blindly trusted -- this part of the program was not written very rigorously. We did not find the optimality important enough to look deeper into this -- it would be much more interesting to find a simple natural reason for regularity than to learn the optimal radius. Currently we do not know how to prove such a thing other than by stumbling upon an inverse. Even if optimality is looked for, it would be more interesting (to the author) to measure the smallest possible Kolmogorov complexity of an inverse rule (in some suitable sense) rather than its radius. See the discussion of presentation of clopen sets below.


\subsection{How were the non-invertibility proofs found?}

For CA not having right inverses, we simply investigated the rule with pen and paper, checking the strong periodic point condition of \cite{SaTo15a}, which was contradicted in each case by points of period one. We again used SAGE to help find the images of the CA.

One gets the feeling that the Strong Law of Small Numbers \cite{Gu88} is at play in the proofs: the patterns $011$, $100$, $110$ and $100$ that break a unary period play a very special role, and they make up half the patterns we can see. (And of course, the other patterns are part of a two-period or a one-period!) With a larger radius, these phenomena become atypical, in the sense that most words do not have anything to do with small periods. Based on this, one might conjecture that regularity is common in cellular automata as radius grows.

\subsection{SFT images}

We needed to compute the images of the 11 ECA, as in the split epic case we needed this to know which words we need to check $fg = \ID$ on. We list the forbidden patterns that we found for SFT cases. They were obtained by computing a DFA for the forbidden words of minimal length, observing it is acyclic and listing the words it accepts. The correctness of these patterns has been checked quite carefully.


\begin{itemize}
\item ECA $6$: $111, 100101, 1001001, 11000101, 110001001$.
st
\item ECA $7$: $1001$.

\item ECA $9$: $1011, 10101, 11001, 11000011, 110000101$.

\item ECA $23$: $0100010, 01001, 01101, 10010, 10110, 1011101$.

\item ECA $28$: $111$.

\item ECA $33$: $1011, 1101, 110011$.

\item ECA $57$: $001000, 111000, 111011, 1110101000$.

\item ECA $77$: $00011, 00111, 11000, 11100, 0001000, 1110111$.
\end{itemize}

\noindent My notes say that ECA $27, 41, 58$ have proper sofic images, but this has not been double checked.

\subsection{Presentation of rules}

The ``first case that applies'' way of representing a rule (used in the Appendix) turned out a very nice was to present weak inverses at least for rules 7, 33 and 77, and allows one to apply the rule manually. The preimage of the cylinder $[1]_0$ under a cellular automaton is a clopen set in the Cantor topology that fully describes it. In our figures we essentially represent such a cylinder as an alternating series
\[ (((... \setminus C_4) \cup C_3) \setminus C_2) \cup C_1 \]
where each $C_i$ is a clopen set in ``disjunctive normal form'', in the sense that it is a union
\[ C_i = D_{i, 1} \cup D_{i, 2} \cup \ldots \cup D_{i, \ell_i} \]
where each
$D_{i, j}$ is an intersection
\[ D_{i, j} = E_{i, j, 1} \cap E_{i, j, 2} \cap \ldots \cap E_{i, j, m_{i, j}} \]
where each $E_{i, j, k}$ is a primitive cylinder $[a]_n$, $a \in \{0,1\}$, $n \in \Z$.

Trying to optimize the number of $C_i$ (the number of alternations of the image bit), the numbers of $D_{i, j}$ (the number of consecutive patterns giving the same bit), and $E_{i,j,k}$ (the numbers of bits specified in a particular pattern) was a lot of fun, but we probably did not optimize them very well.

We also report canonical normal forms in a copy-pasteable form. Choices we considered were decimal (for historical reasons), and bases 16 and 64, but $64 = 2^6$ is not a power tower of $2$ so some structure is lost when using this base. We were surprised by how nice the hexadecimal representation of our inverse for rule 6 looks like in Figure~\ref{fig:rule6inverse}. 
Before trying it out, we had no idea the hexadecimal number would look so regular,\footnote{For comparison, the decimal representation is
\printnumber{26726698204457173985426838654351442812433064955617761561473203597467507804800954948762245864086253298778100885350128103992238993950737242632106437103670939477292080025923981912527590821363309293961298550554954008566332940141141648018729855972959752575589815203856126383648545421716882756091523381830157084652250898696805101034541122348061254604489384065071450200862112014756098843601647635924480904317129640356170620518951557945765409802497497026962812368891989993857798437967198531425023334480163227761193970790742273253443184461707102877712782007076038517756020200420494324263875270861193215}
} as it was found by a (human-guided) brute force computer search, and our presentation of the rule in Figure~\ref{fig:rule6inverse} was quite lengthy. We did not try to find a shorter presentation by drawing this number, as tempting as it sounds.

\section*{Acknowledgements}

We thank Jarkko Kari for observing that Proposition~\ref{prop:SurjectiveNotInjective} works in all dimensions -- my original proof was specific to one dimension. We thank Johan Kopra for pointing out that Lemma~\ref{lem:InjImpliesSurj} is easier to prove than to find in \cite{LiMa95}.

\bibliographystyle{plain}
\bibliography{../../../bib/bib}{}

\newpage

\appendix

\section{Right inverse of ECA 6}

\begin{figure}[h!]
\[ \begin{tikzpicture}[scale = 0.4] \draw (0,0) grid (11,1); \draw (5,1) rectangle (6,2); \node at (6.5,0.5) {1}; \node at (5.5,1.5) {0}; \end{tikzpicture}\;\;\;\;\begin{tikzpicture}[scale = 0.4] \draw (0,0) grid (11,1); \draw (5,1) rectangle (6,2); \node at (3.5,0.5) {1}; \node at (4.5,0.5) {1}; \node at (6.5,0.5) {0}; \node at (5.5,1.5) {0}; \end{tikzpicture} \] 
\[ \begin{tikzpicture}[scale = 0.4] \draw (0,0) grid (11,1); \draw (5,1) rectangle (6,2); \node at (3.5,0.5) {0}; \node at (4.5,0.5) {0}; \node at (5.5,0.5) {0}; \node at (5.5,1.5) {1}; \end{tikzpicture}\;\;\;\;\begin{tikzpicture}[scale = 0.4] \draw (0,0) grid (11,1); \draw (5,1) rectangle (6,2); \node at (4.5,0.5) {1}; \node at (5.5,0.5) {1}; \node at (5.5,1.5) {1}; \end{tikzpicture} \] 
\[ \begin{tikzpicture}[scale = 0.4] \draw (0,0) grid (11,1); \draw (5,1) rectangle (6,2); \node at (4.5,0.5) {1}; \node at (7.5,0.5) {0}; \node at (8.5,0.5) {1}; \node at (5.5,1.5) {1}; \end{tikzpicture}\;\;\;\;\begin{tikzpicture}[scale = 0.4] \draw (0,0) grid (11,1); \draw (5,1) rectangle (6,2); \node at (5.5,0.5) {1}; \node at (7.5,0.5) {1}; \node at (5.5,1.5) {1}; \end{tikzpicture} \] 
\[ \begin{tikzpicture}[scale = 0.4] \draw (0,0) grid (11,1); \draw (5,1) rectangle (6,2); \node at (2.5,0.5) {0}; \node at (3.5,0.5) {1}; \node at (5.5,0.5) {0}; \node at (5.5,1.5) {1}; \end{tikzpicture}\;\;\;\;\begin{tikzpicture}[scale = 0.4] \draw (0,0) grid (11,1); \draw (5,1) rectangle (6,2); \node at (2.5,0.5) {1}; \node at (3.5,0.5) {0}; \node at (5.5,1.5) {0}; \end{tikzpicture} \] 
\[ \begin{tikzpicture}[scale = 0.4] \draw (0,0) grid (11,1); \draw (5,1) rectangle (6,2); \node at (2.5,0.5) {1}; \node at (5.5,0.5) {0}; \node at (7.5,0.5) {0}; \node at (5.5,1.5) {1}; \end{tikzpicture}\;\;\;\;\begin{tikzpicture}[scale = 0.4] \draw (0,0) grid (11,1); \draw (5,1) rectangle (6,2); \node at (5.5,0.5) {1}; \node at (8.5,0.5) {0}; \node at (9.5,0.5) {1}; \node at (5.5,1.5) {0}; \end{tikzpicture} \] 
\[ \begin{tikzpicture}[scale = 0.4] \draw (0,0) grid (11,1); \draw (5,1) rectangle (6,2); \node at (1.5,0.5) {0}; \node at (2.5,0.5) {0}; \node at (4.5,0.5) {0}; \node at (5.5,1.5) {1}; \end{tikzpicture}\;\;\;\;\begin{tikzpicture}[scale = 0.4] \draw (0,0) grid (11,1); \draw (5,1) rectangle (6,2); \node at (0.5,0.5) {1}; \node at (5.5,0.5) {0}; \node at (7.5,0.5) {0}; \node at (5.5,1.5) {1}; \end{tikzpicture} \] 
\[ \begin{tikzpicture}[scale = 0.4] \draw (0,0) grid (11,1); \draw (5,1) rectangle (6,2); \node at (1.5,0.5) {1}; \node at (4.5,0.5) {1}; \node at (5.5,1.5) {1}; \end{tikzpicture}\;\;\;\;\begin{tikzpicture}[scale = 0.4] \draw (0,0) grid (11,1); \draw (5,1) rectangle (6,2); \node at (5.5,0.5) {0}; \node at (5.5,1.5) {0}; \end{tikzpicture} \] 
\[ \begin{tikzpicture}[scale = 0.4]
\draw (0,0) grid (11,1); \draw (5,1) rectangle (6,2);
\node at (3.5,0.5) {0};
\node at (8.5,0.5) {0};
\node at (5.5,1.5) {0};
\end{tikzpicture}\;\;\;\;
\begin{tikzpicture}[scale = 0.4]
\draw (0,0) grid (11,1); \draw (5,1) rectangle (6,2); \node at (5.5,1.5) {1};
\end{tikzpicture} \] 
\caption{A weak generalized inverse of ECA $6$. The rules are applied row by row, and on each row from left to right. An empty box denotes a wildcard symbol, and the first rule to apply is used. The rightmost coordinate is not actually read by any rule. This is the radius 5 binary CA with the hex number
\centerline{\texttt{00000000000000000000FFF3000000FF0000FFFF000000F00000FF000000FFFF}}
\centerline{\texttt{00000000000000000000FFF30000FFFF0000FFFF0000FFFF0000FFF00000FFFF}}
\centerline{\texttt{00000000000000000000FFF3000000FF0000FFFF000000F00000FF000000FFFF}}
\centerline{\texttt{00000000000000000000FFF30000FFFF0000FFFF000000FF0000FFF30000FFFF}}
\centerline{\texttt{00000000000000000000FFF3000000FF0000FFFF000000F00000FF000000FFFF}}
\centerline{\texttt{00000000000000000000FFF30000FFFF0000FFFF0000FFFF0000FFF00000FFFF}}
\centerline{\texttt{00000000000000000000FFF3000000FF0000FFFF000000F00000FF000000FFFF}}
\centerline{\texttt{00000000000000000000FFF30000FFFF0000FFFF000000F00000FFF30000FFFF}}
}
\label{fig:rule6inverse}
\end{figure}
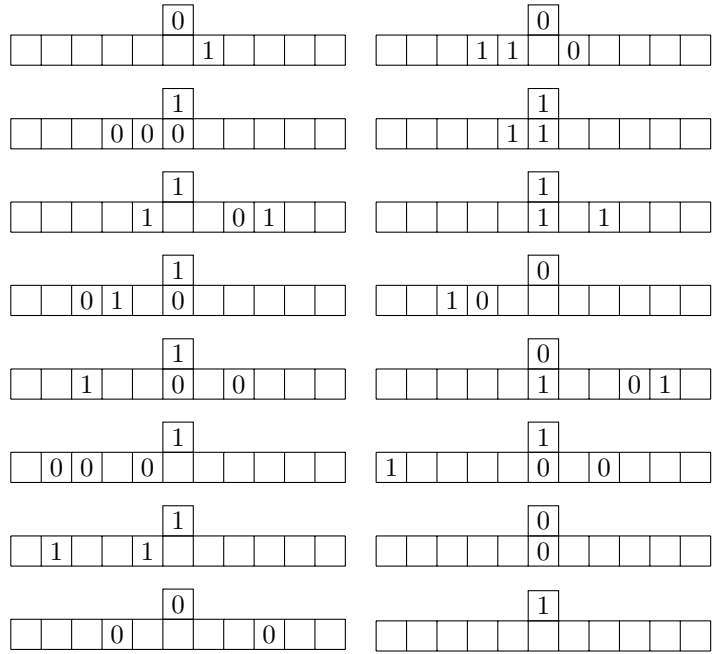 

\newpage

\section{Right inverse of ECA 7}

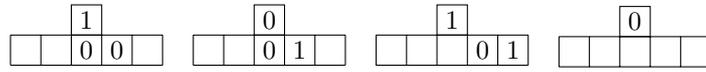
\begin{figure}[h!]
\[ \begin{tikzpicture}[scale = 0.4] \draw (0,0) grid (5,1); \draw (2,1) rectangle (3,2); \node at (2.5,0.5) {0}; \node at (3.5,0.5) {0}; \node at (2.5,1.5) {1}; \end{tikzpicture}\;\;\;\;\begin{tikzpicture}[scale = 0.4] \draw (0,0) grid (5,1); \draw (2,1) rectangle (3,2); \node at (2.5,0.5) {0}; \node at (3.5,0.5) {1}; \node at (2.5,1.5) {0}; \end{tikzpicture}\;\;\;\;\begin{tikzpicture}[scale = 0.4] \draw (0,0) grid (5,1); \draw (2,1) rectangle (3,2); \node at (3.5,0.5) {0}; \node at (4.5,0.5) {1}; \node at (2.5,1.5) {1}; \end{tikzpicture}\;\;\;\;\begin{tikzpicture}[scale = 0.4] \draw (0,0) grid (5,1); \draw (2,1) rectangle (3,2); \node at (2.5,1.5) {0}; \end{tikzpicture} \] 
\caption{A weak generalized inverse of ECA $7$. Radius 2 binary CA with the number \texttt{23232323}, equal to ECA $35$ composed with $\sigma$.}
\label{fig:rule7inverse}
\end{figure} 

\section{Right inverse of ECA 23}

\begin{figure}[h!]

\[ \begin{tikzpicture}[scale = 0.4]
\draw (1,0) grid (6,1); \draw (3,1) rectangle (4,2); \node at (1.5,0.5) {1}; \node at (2.5,0.5) {0}; \node at (3.5,1.5) {1}; \end{tikzpicture}\;\;\;\;
\begin{tikzpicture}[scale = 0.4] \draw (1,0) grid (6,1); \draw (3,1) rectangle (4,2); \node at (3.5,0.5) {1}; \node at (4.5,0.5) {1}; \node at (3.5,1.5) {0}; \end{tikzpicture}\;\;\;\;
\begin{tikzpicture}[scale = 0.4] \draw (1,0) grid (6,1); \draw (3,1) rectangle (4,2); \node at (2.5,0.5) {1}; \node at (4.5,0.5) {1}; \node at (3.5,1.5) {0}; \end{tikzpicture}\;\;\;\;
\begin{tikzpicture}[scale = 0.4] \draw (1,0) grid (6,1); \draw (3,1) rectangle (4,2); \node at (2.5,0.5) {0}; \node at (3.5,0.5) {1}; \node at (3.5,1.5) {1}; \end{tikzpicture}\;\;\;\;
\begin{tikzpicture}[scale = 0.4] \draw (1,0) grid (6,1); \draw (3,1) rectangle (4,2); \node at (4.5,0.5) {1}; \node at (5.5,0.5) {0}; \node at (3.5,1.5) {0}; \end{tikzpicture}\]

\[\begin{tikzpicture}[scale = 0.4] \draw (1,0) grid (6,1); \draw (3,1) rectangle (4,2); \node at (3.5,0.5) {1}; \node at (5.5,0.5) {0}; \node at (3.5,1.5) {0}; \end{tikzpicture} \;\;\;\;
\begin{tikzpicture}[scale = 0.4] \draw (1,0) grid (6,1); \draw (3,1) rectangle (4,2); \node at (1.5,0.5) {1}; \node at (3.5,1.5) {1}; \end{tikzpicture}\;\;\;\;\begin{tikzpicture}[scale = 0.4] \draw (1,0) grid (6,1); \draw (3,1) rectangle (4,2); \node at (2.5,0.5) {1}; \node at (3.5,1.5) {0}; \end{tikzpicture}\;\;\;
\begin{tikzpicture}[scale = 0.4] \draw (1,0) grid (6,1); \draw (3,1) rectangle (4,2); \node at (3.5,1.5) {1}; \end{tikzpicture} \] 

\caption{A weak generalized inverse of ECA $23$. Radius 2 binary CA with the number \texttt{23FF003B}.}
\label{fig:rule23inverse}
\end{figure} 

\section{Right inverse of ECA 33}

\begin{figure}[h!]
\[ \begin{tikzpicture}[scale = 0.4] \draw (3,0) grid (8,1); \draw (5,1) rectangle (6,2); \node at (5.5,0.5) {1}; \node at (5.5,1.5) {0}; \end{tikzpicture}\;\;\;\;
\begin{tikzpicture}[scale = 0.4] \draw (3,0) grid (8,1); \draw (5,1) rectangle (6,2); \node at (4.5,0.5) {0}; \node at (6.5,0.5) {0}; \node at (5.5,1.5) {1}; \end{tikzpicture} \;\;\;\;
\begin{tikzpicture}[scale = 0.4] \draw (3,0) grid (8,1); \draw (5,1) rectangle (6,2); \node at (4.5,0.5) {1}; \node at (6.5,0.5) {1}; \node at (5.5,1.5) {1}; \end{tikzpicture} \]
\[ \begin{tikzpicture}[scale = 0.4] \draw (3,0) grid (8,1); \draw (5,1) rectangle (6,2);
\node at (3.5,0.5) {1};
\node at (4.5,0.5) {1};
\node at (6.5,0.5) {0};
\node at (5.5,1.5) {0}; \end{tikzpicture} \;\;\;\;
\begin{tikzpicture}[scale = 0.4] \draw (3,0) grid (8,1); \draw (5,1) rectangle (6,2);
\node at (4.5,0.5) {0};
\node at (6.5,0.5) {1};
\node at (7.5,0.5) {1};
\node at (5.5,1.5) {0};
\end{tikzpicture}\;\;\;\;
\begin{tikzpicture}[scale = 0.4] \draw (3,0) grid (8,1); \draw (5,1) rectangle (6,2); \node at (5.5,1.5) {1}; \end{tikzpicture} \] 
\caption{A weak generalized inverse of ECA $33$. Radius 2 binary CA with the number \texttt{0C070F07}.}
\label{fig:rule33inverse}
\end{figure}

\newpage

\section{Right inverse of ECA 57}

\begin{figure}[h!]
\[ \begin{tikzpicture}[scale = 0.4] \draw (0,0) grid (9,1); \draw (4,1) rectangle (5,2); \node at (4.5,0.5) {1}; \node at (5.5,0.5) {1}; \node at (4.5,1.5) {0}; \end{tikzpicture}\;\;\;\;\begin{tikzpicture}[scale = 0.4] \draw (0,0) grid (9,1); \draw (4,1) rectangle (5,2); \node at (3.5,0.5) {0}; \node at (4.5,0.5) {0}; \node at (4.5,1.5) {1}; \end{tikzpicture}\;\;\;\;\begin{tikzpicture}[scale = 0.4] \draw (0,0) grid (9,1); \draw (4,1) rectangle (5,2); \node at (5.5,0.5) {1}; \node at (6.5,0.5) {1}; \node at (4.5,1.5) {1}; \end{tikzpicture} \] 
\[ \begin{tikzpicture}[scale = 0.4] \draw (0,0) grid (9,1); \draw (4,1) rectangle (5,2); \node at (2.5,0.5) {1}; \node at (3.5,0.5) {1}; \node at (6.5,0.5) {1}; \node at (4.5,1.5) {0}; \end{tikzpicture}\;\;\;\;\begin{tikzpicture}[scale = 0.4] \draw (0,0) grid (9,1); \draw (4,1) rectangle (5,2); \node at (2.5,0.5) {0}; \node at (3.5,0.5) {0}; \node at (4.5,1.5) {0}; \end{tikzpicture}\;\;\;\;\begin{tikzpicture}[scale = 0.4] \draw (0,0) grid (9,1); \draw (4,1) rectangle (5,2); \node at (4.5,0.5) {0}; \node at (5.5,0.5) {0}; \node at (6.5,0.5) {0}; \node at (4.5,1.5) {1}; \end{tikzpicture} \] 
\[ \begin{tikzpicture}[scale = 0.4] \draw (0,0) grid (9,1); \draw (4,1) rectangle (5,2); \node at (1.5,0.5) {1}; \node at (2.5,0.5) {1}; \node at (3.5,0.5) {1}; \node at (4.5,1.5) {0}; \end{tikzpicture}\;\;\;\;\begin{tikzpicture}[scale = 0.4] \draw (0,0) grid (9,1); \draw (4,1) rectangle (5,2); \node at (1.5,0.5) {1}; \node at (3.5,0.5) {0}; \node at (6.5,0.5) {0}; \node at (4.5,1.5) {1}; \end{tikzpicture}\;\;\;\;\begin{tikzpicture}[scale = 0.4] \draw (0,0) grid (9,1); \draw (4,1) rectangle (5,2); \node at (1.5,0.5) {0}; \node at (4.5,0.5) {0}; \node at (5.5,0.5) {0}; \node at (4.5,1.5) {0}; \end{tikzpicture} \] 
\[ \begin{tikzpicture}[scale = 0.4] \draw (0,0) grid (9,1); \draw (4,1) rectangle (5,2); \node at (3.5,0.5) {1}; \node at (4.5,0.5) {1}; \node at (6.5,0.5) {1}; \node at (7.5,0.5) {1}; \node at (4.5,1.5) {1}; \end{tikzpicture}\;\;\;\;\begin{tikzpicture}[scale = 0.4] \draw (0,0) grid (9,1); \draw (4,1) rectangle (5,2); \node at (2.5,0.5) {0}; \node at (4.5,0.5) {1}; \node at (7.5,0.5) {1}; \node at (4.5,1.5) {0}; \end{tikzpicture}\;\;\;\;\begin{tikzpicture}[scale = 0.4] \draw (0,0) grid (9,1); \draw (4,1) rectangle (5,2); \node at (0.5,0.5) {1}; \node at (1.5,0.5) {0}; \node at (2.5,0.5) {1}; \node at (4.5,0.5) {1}; \node at (6.5,0.5) {1}; \node at (4.5,1.5) {0}; \end{tikzpicture} \] 
\[ \begin{tikzpicture}[scale = 0.4] \draw (0,0) grid (9,1); \draw (4,1) rectangle (5,2); \node at (1.5,0.5) {1}; \node at (2.5,0.5) {0}; \node at (3.5,0.5) {1}; \node at (5.5,0.5) {1}; \node at (7.5,0.5) {1}; \node at (4.5,1.5) {1}; \end{tikzpicture}\;\;\;\;\begin{tikzpicture}[scale = 0.4] \draw (0,0) grid (9,1); \draw (4,1) rectangle (5,2); \node at (2.5,0.5) {0}; \node at (5.5,0.5) {1}; \node at (7.5,0.5) {1}; \node at (8.5,0.5) {1}; \node at (4.5,1.5) {0}; \end{tikzpicture}\;\;\;\;\begin{tikzpicture}[scale = 0.4] \draw (0,0) grid (9,1); \draw (4,1) rectangle (5,2); \node at (5.5,0.5) {1}; \node at (7.5,0.5) {1}; \node at (4.5,1.5) {1}; \end{tikzpicture} \] 
\[ \begin{tikzpicture}[scale = 0.4] \draw (0,0) grid (9,1); \draw (4,1) rectangle (5,2); \node at (5.5,0.5) {1}; \node at (8.5,0.5) {0}; \node at (4.5,1.5) {0}; \end{tikzpicture}\;\;\;\;\begin{tikzpicture}[scale = 0.4] \draw (0,0) grid (9,1); \draw (4,1) rectangle (5,2); \node at (2.5,0.5) {0}; \node at (6.5,0.5) {0}; \node at (4.5,1.5) {1}; \end{tikzpicture}\;\;\;\;\begin{tikzpicture}[scale = 0.4] \draw (0,0) grid (9,1); \draw (4,1) rectangle (5,2); \node at (2.5,0.5) {1}; \node at (4.5,0.5) {1}; \node at (6.5,0.5) {0}; \node at (7.5,0.5) {0}; \node at (4.5,1.5) {1}; \end{tikzpicture} \] 
\[ \begin{tikzpicture}[scale = 0.4] \draw (0,0) grid (9,1); \draw (4,1) rectangle (5,2); \node at (6.5,0.5) {0}; \node at (4.5,1.5) {0}; \end{tikzpicture}\;\;\;\;\begin{tikzpicture}[scale = 0.4] \draw (0,0) grid (9,1); \draw (4,1) rectangle (5,2); \node at (0.5,0.5) {0}; \node at (1.5,0.5) {1}; \node at (2.5,0.5) {1}; \node at (4.5,1.5) {0}; \end{tikzpicture}\;\;\;\;\begin{tikzpicture}[scale = 0.4] \draw (0,0) grid (9,1); \draw (4,1) rectangle (5,2); \node at (2.5,0.5) {1}; \node at (6.5,0.5) {1}; \node at (7.5,0.5) {1}; \node at (4.5,1.5) {1}; \end{tikzpicture} \] 
\[ \begin{tikzpicture}[scale = 0.4] \draw (0,0) grid (9,1); \draw (4,1) rectangle (5,2); \node at (3.5,0.5) {1}; \node at (4.5,0.5) {1}; \node at (6.5,0.5) {1}; \node at (8.5,0.5) {1}; \node at (4.5,1.5) {1}; \end{tikzpicture}\;\;\;\;\begin{tikzpicture}[scale = 0.4] \draw (0,0) grid (9,1); \draw (4,1) rectangle (5,2); \node at (0.5,0.5) {0}; \node at (4.5,1.5) {0}; \end{tikzpicture}\;\;\;\;\begin{tikzpicture}[scale = 0.4] \draw (0,0) grid (9,1); \draw (4,1) rectangle (5,2); \node at (2.5,0.5) {0}; \node at (4.5,0.5) {1}; \node at (4.5,1.5) {0}; \end{tikzpicture} \] 
\[ \begin{tikzpicture}[scale = 0.4] \draw (0,0) grid (9,1); \draw (4,1) rectangle (5,2); \node at (4.5,1.5) {1}; \end{tikzpicture} \] 
\caption{A weak generalized inverse of ECA $57$. The rules are applied row by row, and on each row from left to right -- the first rule to apply is used, an empty box denoting a wildcard symbol. This is the radius 4 binary CA with the number
\centerline{\texttt{0000F00F00FFFFFF00E3FEFF0000FFFF0003FC0F0003FFFF00E3F60F0000FFFF}}
\centerline{\texttt{0000F00F000FFFFF00E3FE0F0000FFFF0003FC0F00C3FFFF00E3F60F0000FFFF}}
}
\label{fig:rule57inverse}
\end{figure}

\section{Right inverse of ECA 77}

\begin{figure}[h!]

\[ \begin{tikzpicture}[scale = 0.4] \draw (3,0) grid (8,1); \draw (5,1) rectangle (6,2); \node at (4.5,0.5) {0}; \node at (6.5,0.5) {0}; \node at (5.5,1.5) {1}; \end{tikzpicture}\;\;\;\;\begin{tikzpicture}[scale = 0.4] \draw (3,0) grid (8,1); \draw (5,1) rectangle (6,2); \node at (4.5,0.5) {1}; \node at (6.5,0.5) {1}; \node at (5.5,1.5) {0}; \end{tikzpicture} \;\;\;\;
\begin{tikzpicture}[scale = 0.4] \draw (3,0) grid (8,1); \draw (5,1) rectangle (6,2); \node at (3.5,0.5) {0}; \node at (7.5,0.5) {0}; \node at (5.5,1.5) {1}; \end{tikzpicture}\]
\[ \begin{tikzpicture}[scale = 0.4] \draw (3,0) grid (8,1); \draw (5,1) rectangle (6,2); \node at (3.5,0.5) {1}; \node at (7.5,0.5) {1}; \node at (5.5,1.5) {0}; \end{tikzpicture} \;\;\;\;
 \begin{tikzpicture}[scale = 0.4] \draw (3,0) grid (8,1); \draw (5,1) rectangle (6,2); \node at (5.5,0.5) {0}; \node at (5.5,1.5) {0}; \end{tikzpicture}\;\;\;\;\begin{tikzpicture}[scale = 0.4] \draw (3,0) grid (8,1); \draw (5,1) rectangle (6,2); \node at (5.5,0.5) {1}; \node at (5.5,1.5) {1}; \end{tikzpicture} \] 

\caption{A weak generalized inverse of ECA $77$. Radius 5 binary CA with the number \texttt{107331F7}.}
\label{fig:rule77inverse}
\end{figure}

\section{Code for checking the rules}
\label{sec:Code}

\begin{lstlisting}[language=Python]
def r6inv(w):
    if w[6] == "1":
        return "0"
    if w[3] == "1" and w[4] == "1" and w[6] == "0":
        return "0"
    if w[3:6] == "000":
        return "1"
    if w[4:6] == "11":
        return "1"
    if w[4] == "1" and w[7:9] == "01":
        return "1"
    if w[5] == "1" and w[7] == "1":
        return "1"
    if w[2:4] == "01" and w[5] == "0":
        return "1"
    if w[2:4] == "10":
        return "0"
    if w[2] == "1" and w[5] == "0" and w[7] == "0":
        return "1"
    if w[5] == "1" and w[8:10] == "01":
        return "0"
    if w[1:3] == "00" and w[4] == "0":
        return "1"
    if w[0] == "1" and w[5] == "0" and w[7] == "0":
        return "1"
    if w[1] == "1" and w[4] == "1":
        return "1"
    if w[5] == "0":
        return "0"
    if w[3] == "0" and w[8] == "0":
        return "0"
    return "1"

def r7inv(w):
    if w[2] == w[3] == "0":
        return "1"
    if w[2] == "0" and w[3] == "1":
        return "0"
    if w[3] == "0" and w[4] == "1":
        return "1"
    return "0"

def r23inv(w):
    if w[0:2] == "10":
        return "1"
    if w[2:4] == "11":
        return "0"
    if w[1] == w[3] == "1":
        return "0"
    if w[1:3] == "01":
        return "1"
    if w[3:5] == "10":
        return "0"
    if w[2] == "1" and w[4] == "0":
        return "0"
    if w[0] == "1":
        return "1"
    if w[1] == "1":
        return "0"
    return "1"

def r33inv(w):
    if w[2] == "1":
        return "0"
    if w[1] == w[3] == "0":
        return "1"
    if w[1] == w[3] == "1":
        return "1"
    if w[0] == w[1] == "1" and w[3] == "0":
        return "0"
    if w[1] == "0" and w[3] == w[4] == "1":
        return "0"
    return "1"

def r57inv(w):
    if w[4] == w[5] == "1":
        return "0"
    if w[3] == w[4] == "0":
        return "1"
    if w[5] == w[6] == "1":
        return "1"

    if w[2] == w[3] == w[6] == "1":
        return "0"
    if w[2] == w[3] == "0":
        return "0"
    if w[4] == w[5] == w[6] == "0":
        return "1"

    if w[1] == w[2] == w[3] == "1":
        return "0"
    if w[1] == "1" and w[3] == "0" and w[6] == "0":
        return "1"
    if w[1] == "0" and w[4] == "0" and w[5] == "0":
        return "0"

    if w[3] == w[4] == w[6] == w[7] == "1":
        return "1"
    if w[2] == "0" and w[4] == "1" and w[7] == "1":
        return "0"
    if w[0] == "1" and w[1] == "0" and w[2] == w[4] == w[6] == "1":
        return "0"

    if w[1] == "1" and w[2] == "0" and w[3] == w[5] == w[7] == "1":
        return "1"
    if w[2] == "0" and w[5] == "1" and w[7] == w[8] == "1":
        return "0"
    if w[5] == w[7] == "1":
        return "1"

    if w[5] == "1" and w[8] == "0":
        return "0"
    if w[2] == "0" and w[6] == "0":
        return "1"
    if w[2] == w[4] == "1" and w[6] == "0" and w[7] == "0":
        return "1"

    if w[6] == "0":
        return "0"
    if w[0] == "0" and w[1] == "1" and w[2] == "1":
        return "0"
    if w[2] == "1" and w[6] == w[7] == "1":
        return "1"

    if w[3] == w[4] == w[6] == w[8] == "1":
        return "1"
    if w[0] == "0":
        return "0"
    if w[2] == "0" and w[4] == "1":
        return "0"

    return "1"

def r77inv(w):
    if w[1] == w[3] == "0":
        return "1"
    if w[1] == w[3] == "1":
        return "0"
    if w[0] == w[4] == "0":
        return "1"
    if w[0] == w[4] == "1":
        return "0"
    return w[2]

def r6(w):
    if w == "001" or w == "010":
        return "1"
    return "0"

def r7(w):
    if w == "000" or w == "001" or w == "010":
        return "1"
    return "0"

def r23(w):
    if w == "100" or w == "010" or w == "001" or w == "000":
        return "1"
    return "0"

def r33(w):
    if w == "101" or w == "000":
        return "1"
    return "0"

# 57 = 32 + 16 + 8 + 1 = 101 100 011 000
def r57(w):
    if w == "101" or w == "100" or w == "011" or w == "000":
        return "1"
    return "0"

# 77 = 64 + 13 = 64 + 8 + 4 + 1 = 110 011 010 000
def r77(w):
    if w == "110" or w == "011" or w == "010" or w == "000":
        return "1"
    return "0"
    
def appf(w, f, n):
    r = ""
    for i in range(len(w)-n+1):
        r += f(w[i:i+n])
    return r

def words(n):
    if n == 0:
        yield ""
        return
    for w in words(n-1):
        yield w + "0"
        yield w + "1"

fr6 = lambda a:appf(a, r6, 3)
fr7 = lambda a:appf(a, r7, 3)
fr23 = lambda a:appf(a, r23, 3)
fr33 = lambda a:appf(a, r33, 3)
fr57 = lambda a:appf(a, r57, 3)
fr77 = lambda a:appf(a, r77, 3)
fr6inv = lambda a:appf(a, r6inv, 11)
fr7inv = lambda a:appf(a, r7inv, 5)
fr23inv = lambda a:appf(a, r23inv, 5)
fr33inv = lambda a:appf(a, r33inv, 5)
fr57inv = lambda a:appf(a, r57inv, 9)
fr77inv = lambda a:appf(a, r77inv, 5)

# testing rule 6
for w in words(15):
    if fr6(w)[6] != fr6(fr6inv(fr6(w))):
        print ("problem for rule 6:")
        print (w)
        print (" "+fr6(w))
        print ("      "+fr6inv(fr6(w)))
        print ("       "+fr6(fr6inv(fr6(w))))
        break
else:
    print ("inverse of rule 6 is ok")

# testing rule 7
for w in words(9):
    if fr7(w)[3] != fr7(fr7inv(fr7(w))):
        print ("problem with rule 7:")
        print (w)
        print (" "+fr7(w))
        print ("   "+fr7inv(fr7(w)))
        print ("    "+fr7(fr7inv(fr7(w))))
        break
else:
    print ("inverse of rule 7 is ok")

# testing rule 23
for w in words(9):
    if fr23(w)[3] != fr23(fr23inv(fr23(w))):
        print ("problem with rule 23:")
        print (w)
        print (" "+fr23(w))
        print ("   "+fr23inv(fr23(w)))
        print ("    "+fr23(fr23inv(fr23(w))))
        break
else:
    print ("inverse of rule 23 is ok")

# testing rule 33
for w in words(9):
    if fr33(w)[3] != fr33(fr33inv(fr33(w))):
        print ("problem with rule 33:")
        print (w)
        print (" "+fr33(w))
        print ("   "+fr33inv(fr33(w)))
        print ("    "+fr33(fr33inv(fr33(w))))
        break
else:
    print ("inverse of rule 33 is ok")

# testing rule 57
for w in words(13):
    if fr57(w)[5] != fr57(fr57inv(fr57(w))):
        print ("problem for rule 57:")
        print (w)
        print (" "+fr57(w))
        print ("     "+fr57inv(fr57(w)))
        print ("      "+fr57(fr57inv(fr57(w))))
        break
else:
    print ("inverse of rule 57 is ok")

# testing rule 77
for w in words(9):
    if fr77(w)[3] != fr77(fr77inv(fr77(w))):
        print ("problem with rule 77:")
        print (w)
        print (" "+fr77(w))
        print ("   "+fr77inv(fr77(w)))
        print ("    "+fr77(fr77inv(fr77(w))))
        break
else:
    print ("inverse of rule 77 is ok")
\end{lstlisting}

\section{Code for checking the numbers}
\label{sec:CodeNum}

\begin{lstlisting}[language=Python]
# we use the convention that for radius 1 numbers are in decimal,
# and in hex from nbhd size 4 onward

r6invnum = """
00000000000000000000FFF3000000FF0000FFFF000000F00000FF000000FFFF
00000000000000000000FFF30000FFFF0000FFFF0000FFFF0000FFF00000FFFF
00000000000000000000FFF3000000FF0000FFFF000000F00000FF000000FFFF
00000000000000000000FFF30000FFFF0000FFFF000000FF0000FFF30000FFFF
00000000000000000000FFF3000000FF0000FFFF000000F00000FF000000FFFF
00000000000000000000FFF30000FFFF0000FFFF0000FFFF0000FFF00000FFFF
00000000000000000000FFF3000000FF0000FFFF000000F00000FF000000FFFF
00000000000000000000FFF30000FFFF0000FFFF000000F00000FFF30000FFFF
"""

r7invnum = "23232323"

r23invnum = "23FF003B"

r33invnum = "0C070F07"

r57invnum = """
0000F00F00FFFFFF00E3FEFF0000FFFF0003FC0F0003FFFF00E3F60F0000FFFF
0000F00F000FFFFF00E3FE0F0000FFFF0003FC0F00C3FFFF00E3F60F0000FFFF
"""

r77invnum = "107331F7"

r6num = "6"
r7num = "7"
r23num = "23"
r33num = "33"
r57num = "57"
r77num = "77"

def singleappfnum(numstr, n, w):
    numstr = ''.join(numstr.split())
    if len(w) != n:
        print ("Wrong neighborhood.")
    if n <= 3:
        num = int(numstr)
    else:
        num = int(numstr, 16)
    # now convert to binary and remove "0b" prefix
    binrepr = bin(num)[2:]
    widx = int(w, 2)
    if widx >= len(binrepr):
        return "0"
    return binrepr[-widx-1]
    
def appfnum(w, numstr, n):
    r = ""
    for i in range(len(w)-n+1):
        # NB. I know it's reparsed every time.
        r += singleappfnum(numstr, n, w[i:i+n])
    return r

def words(n):
    if n == 0:
        yield ""
        return
    for w in words(n-1):
        yield w + "0"
        yield w + "1"

fr6 = lambda a:appfnum(a, r6num, 3)
fr7 = lambda a:appfnum(a, r7num, 3)
fr23 = lambda a:appfnum(a, r23num, 3)
fr33 = lambda a:appfnum(a, r33num, 3)
fr57 = lambda a:appfnum(a, r57num, 3)
fr77 = lambda a:appfnum(a, r77num, 3)
fr6inv = lambda a:appfnum(a, r6invnum, 11)
fr7inv = lambda a:appfnum(a, r7invnum, 5)
fr23inv = lambda a:appfnum(a, r23invnum, 5)
fr33inv = lambda a:appfnum(a, r33invnum, 5)
fr57inv = lambda a:appfnum(a, r57invnum, 9)
fr77inv = lambda a:appfnum(a, r77invnum, 5)

# testing rule 6
for w in words(15):
    if fr6(w)[6] != fr6(fr6inv(fr6(w))):
        print ("problem for rule 6:")
        print (w)
        print (" "+fr6(w))
        print ("      "+fr6inv(fr6(w)))
        print ("       "+fr6(fr6inv(fr6(w))))
        break
else:
    print ("inverse of rule 6 is ok")

# testing rule 7
for w in words(9):
    if fr7(w)[3] != fr7(fr7inv(fr7(w))):
        print ("problem with rule 7:")
        print (w)
        print (" "+fr7(w))
        print ("   "+fr7inv(fr7(w)))
        print ("    "+fr7(fr7inv(fr7(w))))
        break
else:
    print ("inverse of rule 7 is ok")

# testing rule 23
for w in words(9):
    if fr23(w)[3] != fr23(fr23inv(fr23(w))):
        print ("problem with rule 23:")
        print (w)
        print (" "+fr23(w))
        print ("   "+fr23inv(fr23(w)))
        print ("    "+fr23(fr23inv(fr23(w))))
        break
else:
    print ("inverse of rule 23 is ok")

# testing rule 33
for w in words(9):
    if fr33(w)[3] != fr33(fr33inv(fr33(w))):
        print ("problem with rule 33:")
        print (w)
        print (" "+fr33(w))
        print ("   "+fr33inv(fr33(w)))
        print ("    "+fr33(fr33inv(fr33(w))))
        break
else:
    print ("inverse of rule 33 is ok")

# testing rule 57
for w in words(13):
    if fr57(w)[5] != fr57(fr57inv(fr57(w))):
        print ("problem for rule 57:")
        print (w)
        print (" "+fr57(w))
        print ("     "+fr57inv(fr57(w)))
        print ("      "+fr57(fr57inv(fr57(w))))
        break
else:
    print ("inverse of rule 57 is ok")

# testing rule 77
for w in words(9):
    if fr77(w)[3] != fr77(fr77inv(fr77(w))):
        print ("problem with rule 77:")
        print (w)
        print (" "+fr77(w))
        print ("   "+fr77inv(fr77(w)))
        print ("    "+fr77(fr77inv(fr77(w))))
        break
else:
    print ("inverse of rule 77 is ok")
\end{lstlisting}

\end{document}